\documentclass[submission,copyright]{eptcs}
\providecommand{\event}{WORDS 2011}
\usepackage[english]{babel}
\usepackage{breakurl}             
\usepackage {array} 
\usepackage{gastex}
\usepackage{amssymb,amsfonts,amsmath,amsthm}
\newtheorem {teo} {Theorem}
\newtheorem {exmp} {Example}
\newtheorem {Rmk} {Remark}
\newtheorem {Lemma} {Lemma}
\DeclareMathOperator{\per}{per}
\DeclareMathOperator{\lexp}{lexp}
\DeclareMathOperator{\rev}{rev}
\def\CF{{\sf CF}}
\def\OF{{\sf OF}}
\def\TM{{\sf TM}}
\def\U{_\infty^{\phantom{i}a}T}

\title{Constructing Premaximal Binary Cube-free Words\\ of Any Level}
\author{Elena A. Petrova
\institute{Ural Federal University\\
Ekaterinburg, Russia}
\email{captain@akado-ural.ru}
\and Arseny M. Shur
\institute{Ural Federal University\\
Ekaterinburg, Russia}
\email{arseny.shur@usu.ru}
}
\def\titlerunning{Premaximal Binary Cube-free Words}
\def\authorrunning{E. A. Petrova, A. M. Shur}
\begin{document}
\maketitle

\begin{abstract}
We study the structure of the language of binary cube-free words. Namely, we are interested in the cube-free words that cannot be infinitely extended preserving cube-freeness. We show the existence of such words with arbitrarily long finite extensions, both to one side and to both sides. 
\end{abstract}

\section{Introduction}
The study of repetition-free words and languages remains quite popular in combinatorics of words: lots of interesting and challenging problems are still open. The most popular repetition-free binary languages are the \textit{cube-free} language $\CF$ and the \textit{overlap-free} language $\OF$. The language $\CF$ is much bigger and has much more complicated structure. For example, the number of overlap-free binary words grows only polynomially with the length \cite{RS}, while the language of cube-free words has exponential growth \cite{Bra}. The most accurate bounds for the growth of $\OF$ is given in \cite{JPB} and for the growth of $\CF$ in \cite{Sh09dlt}. Further, there is essentially unique nontrivial morphism preserving $\OF$ \cite{See}, while there are uniform morphisms of any length preserving $\CF$ \cite{CR}. The sets of two-sided infinite overlap-free and cube-free binary words also have quite different structure, see \cite{Sh_r00}.

Any repetition-free language can be viewed as a poset with respect to prefix, suffix, or factor order. In case of prefix [suffix] order, the diagram of such a poset is a tree; each node generates a subtree and is a common prefix [respectively, suffix] of its descendants. The following questions arise naturally. \textit{Does a given word generate finite or infinite subtree? Are the subtrees generated by two given words isomorphic? Can words generate arbitrarily large finite subtrees?} For some power-free languages, the decidability of the first question was proved in \cite{Cur} as a corollary of interesting structural properties. The third question for ternary square-free words constitutes Problem~1.10.9 of \cite{AS}. For all $k$th power-free languages, it was shown in \cite{BEM} that the subtree generated by any word has at least one leaf. Note that considering the factor order instead of the prefix or the suffix one, we get a more general acyclic graph instead of a tree, but still can ask the same questions about the structure of this graph. For the language $\OF$, all these questions were answered in \cite{Sh_r98,Sh11sf}, but almost nothing is known about the same questions for $\CF$.

In this paper, we answer the third question for the language $\CF$ in the affirmative. Namely, we construct cube-free words that generate subtrees of any prescribed depth and then extend this result for the subgraphs of the diagram of factor order. 

\section {Preliminaries}
Let us recall necessary notation and definitions. We consider finite and infinite words over the binary alphabet $\Sigma=\{a,b\}$. If $x$ is a letter, then $\bar{x}$ denotes the other letter. By default, ``word'' means a finite word. Words are denoted by uppercase characters (to denote one-sided infinite words, we add the subcsript $_\infty$ at the corresponding side). We write $\lambda$ for the \textit{empty word}, and $|W|$ for the length of the word $W$. The letters of nonempty finite and right-infinite words are numbered from 1; thus, $W=W(1)W(2)\cdots W(|W|)$. The letters of left-infinite words are numbered by \textit{all nonnegative integers}, starting from the right.

We use standard definitions of factors, prefixes, and suffixes of a word. The factor $W(i)\cdots W(j)$ is written as $W(i\ldots j)$. A positive integer $p\leq |W|$ is a \textit{period} of a word $W$ if $W(i)=W(i{+}p)$ for all $i\in \{1,\ldots,|W|{-}p\}$. The minimal period of $W$ is denoted by $\per(W)$. The \textit{exponent} of a word is the ratio between its length and its minimal period: $\exp(W)=|W|/\per(W)$. Words of exponent 2 and 3 are called squares and cubes, respectively. The \textit{local exponent} of a word is the number $\lexp(W)=\sup\{\exp(V)\vert V\text{ is a factor of } W\}$. Periodic words possess the \textit{interaction property} expressed by the textbook Fine and Wilf theorem: if a word $U$ has periods $p$ and $q$, and $|U|\ge p+q-\gcd(p,q)$, then $U$ has the period $\gcd(p,q)$.

A word $W$ is $\beta$-free [$\beta^+\!$-free] if $\lexp(W)<\beta$ [respectively, $\lexp(W)\le \beta$]. The 3-free words are called \textit{cube-free}, and the $2^+\!$-free words are \textit{overlap-free}. The language of all cube-free [overlap-free] words over $\Sigma$ is denoted by $\CF$ [respectively, $\OF$]. 
A morphism $f:\Sigma^+\to\Sigma^+$ \textit{avoids an exponent} $\beta$ if the condition $\lexp(U)<\beta$ implies $\lexp(f(U))<\beta$ for any word $U$. The following theorem allowes one to check cube-freeness of a morphism over the binary alphabet.

\begin {teo}[\cite{RW}] {\label {morph}}
A morphism $f: \Sigma^+ \rightarrow \Sigma^+$ is cube-free if and only if the word \newline $f(aabbababbabbaabaababaabb)$ is cube-free.
\end {teo}

The \textit{Thue--Morse morphism} $\theta$ is defined over $\Sigma^+$ by the rules $\theta(a)=ab$, $\theta(b)=ba$. The words 
$$
T_n^a=\theta^n(a),\ T_n^b=\theta^n(b)\ (n\geq0)
$$
are called \textit{Thue--Morse blocks} or simply $n$-blocks. From the definition it follows that $T_{n{+}1}^x=T_n^xT_n^{\bar{x}}$.
Hence, the sequences $\{T_n^a\}$ and $\{T_n^b\}$ have ``limits'', which are right-infinite \textit{Thue-Morse words} $T_\infty^a$ and $T_\infty^b$, respectively. We also consider the reversal $\U$ of $T_\infty^a$. The factors of Thue-Morse words are \textit{Thue-Morse factors}; the set of all these factors is denoted by $\TM$. Note that any word in $\TM$ can be written as $W=xQ_1\cdots Q_ny$, where $x,y\in\Sigma\cup\{\lambda\}$, $Q_1,\ldots,Q_n\in\{ba, ab\}$. It is known since Thue \cite{Th12} that $\TM\subset\OF$.

Let $L\subset\Sigma^*$ and $W\in L$. Any word $U\in\Sigma^*$ such that $UW\in L$ is called a \textit{left context} of $W$ in $L$. The word $W$ is \textit{left maximal} [\textit{left premaximal}] if it has no nonempty left contexts [respectively, finitely many left contexts]. The \textit{level} of the left premaximal word $W$ is the length of its longest left context; thus, left maximal words are of level 0. The right counterparts of the above notions are defined in a symmetric way. We say that a word is \textit{maximal [premaximal]} if it is both left and right maximal [respectively, premaximal]. The \textit{level} of a premaximal word $W$ is the pair $(n,k)\in \mathbb{N}$ such that $n$ and $k$ are the length of the longest left context of $W$ and the length of its longest right context, respectively.

In particular, a word $W\in\CF$ is maximal if by adding any of the two letters on the left or on the right we obtain a cube. The word $aabaabaa$ is an example of such a word.

The aim of this paper is to prove the following theorems:

\begin {teo} \label{left}
In $\CF$, there exist left premaximal words of any level $n\in\mathbb{N}_0$.
\end {teo}

\begin {teo}\label{lr}
In $\CF$, there exist premaximal words of any level $(n,k)\in\mathbb{N}_0^2$.
\end {teo}

\section {Construction of premaximal words}

Theorem~\ref{left} is proved by exhibiting a series of left premaximal words, containing words of any level. The series is constructed in two steps:
\begin {enumerate}
\item building an auxiliary series $\{W_n\}_0^{\infty}$ such that each word $W_n$ has, up to one easily handled exception, a unique left context of any length $\le n$;
\item completing the word $W_n$ to a left premaximal word $\overline{W}_n$.
\end {enumerate}

If a word $W\in\CF$ has a unique left context of length $n$, say $U$, and two left contexts of length $n{+}1$, then we say that $U$ is the \textit{fixed} left context of $W$ (see the picture below). 

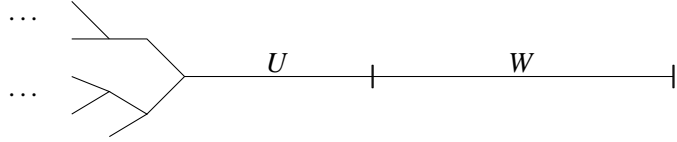
\centerline{ 
\begin{picture}(90,21)(0,2)
\gasset{Nw=0.15,Nh=3,Nmr=0,AHnb=0}
\drawline(10,20)(15,15)
\drawline(10,15)(20,15)(25,10)(90,10)
\node(1)(50,10){}
\node(2)(90,10){}
\drawline(10,10)(15,8)(10,5)
\drawline(15,8)(20,5)(15,2)
\drawline(20,5)(25,10)
\put(37.5,10.5){\makebox(0,0)[cb]{$U$}}
\put(70,10.5){\makebox(0,0)[cb]{$W$}}
\put(6,7.5){\makebox(0,0)[rc]{$\cdots$}}
\put(6,17.5){\makebox(0,0)[rc]{$\cdots$}}
\end{picture}
}

\begin{exmp} \label{abb}
Let $W=aabaaba$. Since $aW=aaa\cdots$, $abW=(aba)^3$, but $aabbW,babbW\in\CF$, we see that the fixed left context of the word $W$ equals $abb$.
\end{exmp}

Now let us explain step 1. We build the series $\{W_n\}_0^{\infty}$ inductively, one word per iteration, in a way that the fixed left context $X_n$ of the word $W_n$ is of length $\ge n$ (we will discuss the mentioned exception at the moment of its appearance). We put $W_0=aabaaba$ and note that the left-infinite word
$$
\U\, abaaba =\cdots abba\,baab\,baab\,abbW_0
$$
is cube-free. So, we require that each word $W_n$ satisfies the following properties:
\begin{itemize}
\item[(W1)] $W_n$ starts with $W_0$;
\item[(W2)] any word $\U(k\ldots1)$ is a left context of $W_n$;
\item[(W3)] some word $\U(k\ldots1)$ with $k\ge n$ is the fixed left context of $W_n$, denoted by $X_n$;
\item[(W4)] if $|X_n|>n$, then $W_{n{+}1}=W_n$ (\textit{trivial} iterations). 
\end{itemize}
The basic idea for obtaining $W_{n+1}$ from $W_n$ at nontrivial iterations is to let
\begin {equation} \label{eq:nos}
W_{n+1} = \underbrace{\phantom{xX_n}W_n}  \underbrace{xX_nW_n}\underbrace{ xX_nW_n},
\end {equation}
where $x$ is the letter ``prohibited'' at the $(n{+}1)$th iteration, i.e. $xX_n$ certainly is not a left context of $W_{n+1}$. Thus, the fixed left context of $W_{n{+}1}$ is longer than the one of $W_n$ by definition.
\begin {Rmk}
An attempt to build the series $\{W_n\}_0^{\infty}$ directly by (\ref{eq:nos}) fails because cubes will occur at the border of some words $W_n$ and $xX_n$. For instance, let us construct the word $W_4$. We have $W_3=W_0$ in view of (W4) and Example~\ref{abb}, $X_3=abb$, and the context $aabb$ should be forbidden in view of (W2), because $\U(4\ldots1)=babb$. So, $x=a$ and the word $W_3xX_3$ has the factor $aaa$.
\end {Rmk}
A way out from this situation is the following idea: we insert a special ``buffer'' word after each of three occurrences of $W_n$ in (\ref{eq:nos}). This insertion allows us to avoid local cubes at the border. Below we use the following notation:
\begin {itemize}
\item[-] $P'_n=xX_n$, $P_n=\bar{x}X_n$, where $x$ is the letter, prohibited at the $(n{+}1)$th iteration; thus, $P_n\in\TM$;
\item[-] $S_n$ is the word inserted after $W_n$ at the $(n{+}1)$th iteration;
\item[-] $S'_n=S_0S_1\cdots S_n$ is the factor of $W_{n{+}1}$ between $W_0$ and the nearest occurrence of $P'_n$;
\item[-] $W'_n=P'_nW_nS_n$.
\end {itemize}
In these terms, we have the following expressions for $W_{n+1}$ for any nontrivial iteration:
\begin{subequations}\label{eq:}
\begin{align}
W_{n+1}&=\underbrace{\phantom{xX_n}W_nS_n}\underbrace{xX_nW_nS_n}\underbrace{xX_nW_nS_n}\label{eq:s}\\
W_{n+1}&=\underbrace{\phantom{P'_n}W_nS_n}\underbrace{ P'_nW_nS_n}\underbrace{P'_nW_nS_n}\label{eq:s1}
\end{align}
\end{subequations}
The structure of the word $W_{n+1}$ imposes the following restrictions on the words $S_n$ and $S_{n+1}$:
\begin {itemize}
\item[(S1)] Since the word $X_{n{+}1}W_{n{+}1}S_{n{+}1}$ is a factor of $W_{n{+}2}$, $X_{n{+}1}$ ends with $X_n$, and $X_nW_{n+1}x=(X_nW_nS_nx)^3$ by (\ref{eq:s}), the word $S_{n+1}$ must start with $\bar{x}$, which is the first letter of $P_n$;
\item[(S2)] Since the word $S_nxX_n$ is a factor of $W_{n{+}1}$, if $X_n$ starts with $x$ [$\bar{x}x\bar{x}x$], then $S_n$ ends with $\bar{x}$ [respectively, $x$]. (Recall that $X_n\in\TM$ is an overlap-free word, whence any other prefix of $X_n$ does not restrict the last letter of $S_n$.)
\end {itemize}
Thus, our first goal is to find the words $S_n$ satisfying (S1) and (S2) such that all words $S'_n$ are cube-free. In other words, we have to construct a cube-free right-infinite word $S'_{\infty}=S_0S_1\cdots S_n\cdots$. The following lemma is easy.

\begin {Lemma} {\label {l:sqr}}
The letters $\U(n)$ and $\U(n{-}1)$ coincide if and only if $n=m\cdot2^k$ for some odd integers $m$ and $k$. 
\end {Lemma}

\begin {Rmk} \label {rm:free}
If the only left context of length $n$ of the word $W_n$ begins with $xx$, then $|X_n|>n$, because the letter before $xx$ is also fixed. Thus, by \emph{(W4)} we have $W_{n{+}1}=W_n$ (and then $S_n=\lambda$) for all values of $n$ mentioned in Lemma~\ref{l:sqr}. For all other values of $n$ ($n>3$), the iterations will be nontrivial.
\end {Rmk}

While constructing the word  $S'_{\infty}$ we follow the next four rules:
\begin {enumerate}
\item For all nontrivial iterations, $S_n\in\{T^x_2,T^x_2T^x_2,T^x_4,T^x_2T^{\overline{x}}_2T^x_1,T^x_1,T^x_1T^{\bar{x}}_2\,|\,x\in\Sigma\}$; hence, $S_n \in \TM$.
\item Whenever possible, we choose $S_n$ to be a 2-block or a product of 2-blocks.
\item Otherwise, if $S_n$ ends with the block $T^x_1$, we put $S_{n+1}=T^{\bar{x}}_1$ or $S_{n+1}=T^{\bar{x}}_1T_2^x$ (or the same possibilities for $S_{n+2}$ if $S_{n{+}1}=\lambda$).
\item If $S_n\ne\lambda$ and there is no restriction (S2) on the last letter of $S_n$, we add this restriction artificially. Namely, we fix the last letter of $S_n$ to be $\bar{x}$ if $S_{n{-}1}$ ends with $x$ (or if $S_{n{-}2}$ ends with $x$ while $S_{n{-}1}=\lambda$). 
\end {enumerate}

Taking rules 1--4 into account, we can prove, by case examination, the following lemma about the first and the last letters of the words $S_n$.

\begin {Lemma} {\label {lcher}}
\emph{(1)} If $S_n$ ends with $x$, then either $S_{n{+}1}$ ends with $\bar{x}$, or $S_{n{+}1}=\lambda$ and $S_{n{+}2}$ ends with $\bar{x}$.\\
\emph{(2)} The first letter of a nonempty word $S_n$ coincides with the last one for all $n$, except for the cases when $P_n=x\bar{x}x\bar{x}\cdots$ or $P_n=xx\bar{x}x\cdots$.
\end {Lemma}

\setlength {\extrarowheight}{0.8pt}
\tabcolsep=3pt
\begin {table}[!htb]
\vspace*{-3mm}
\caption {the suffixes $S_n$ for 32 successive iterations starting from some number $k$ divisible by 32. The righthand [lefthand] part of the table applies if the current letter of $T^b_{\infty}$ is equal [resp., not equal] to the previous one. Trivial iterations are omitted.}
\label {tab1}
\begin {center}
\begin {tabular}  {|l|l|l|l|} 
\hline
Iteration no.  &\multicolumn{2}{c|}{Prohibitions} &  \\
$(n)$ &Start &End &$S_{n-1}$ \\

\hline
$k$ & $\overline{x}$ & $\overline{x}$ & $T^{x}_2$  \\ \hline
$k+1$ &&& \\ \hline

$k+2$ & $x$ & $x$ & $T^{\overline{x}}_2T^{\overline{x}}_2$  \\ \hline

$k+4$ & $\overline{x}$ & $\overline{x}$ & $T^{x}_2$\\ \hline

$k+5$ & $\overline{x}$ & $x, T^{\overline{x}}_2$ & $T^{x}_2T^{\overline{x}}_2T^{x}_1$  \\ \hline

$k+6$ & $x$ & $\overline{x}$ & $T^{\overline{x}}_1$  \\ \hline

$k+8$ & $x$ & $x$ & $T^{\overline{x}}_2$  \\ \hline

$k+10$ & $\overline{x}$ & $\overline{x}$ & $T^{x}_2T^{x}_2$  \\ \hline

$k+12$ & $x$ & $x$ & $T^{\overline{x}}_2$ \\ \hline

$k+13$ & $x$ & $\overline{x}, T^{x}_2$ & $T^{\overline{x}}_2T^{x}_2T^{\overline{x}}_1$ \\ \hline

$k+14$ & $\overline{x}$ & $x$ & $T^{x}_1$ \\ \hline

$k+16$ & $\overline{x}$ & $\overline{x}$ & $T^{x}_2$ \\ \hline

$k+17$ & $\overline{x}$ & $x$ & $T^{x}_1$ \\ \hline

$k+18$ & $xx\overline{x}$ & $\overline{x}$ & $T^{\overline{x}}_1$ \\ \hline

$k+20$ & $\overline{x}\overline{x}x$ & $x$& $T^{\overline{x}}_2$\\ \hline

$k+21$ & $x$ & $\overline{x}$ & $T^{\overline{x}}_1$ \\ \hline

$k+22$ & $\overline{x}\overline{x}x$ & $x$ & $T^{x}_1$ \\ \hline

$k+24$ & $xx\overline{x}$ & $\overline{x}$ & $T^{x}_2$ \\ \hline

$k+26$ & $x$ & $x$ & $T^{\overline{x}}_2$ \\ \hline

$k+28$ & $\overline{x}$ & $\overline{x}$ & $T^{x}_4$ \\ \hline

$k+29$ & $\overline{x}$ & $x, T^{\overline{x}}_2$ & $T^{x}_2T^{\overline{x}}_2T^{x}_1$ \\ \hline

$k+30$ & $x$ & $\overline{x}$ & $T^{\overline{x}}_1 (T^{\overline{x}}_1T^{x}_2)$ \\ \hline
\end {tabular}
\hspace{15 pt}
\begin {tabular}  {|l|l|l|l|} 
\hline
Iteration no.  &\multicolumn{2}{c|}{Prohibitions} &  \\
$(n)$ &Start &End &$S_{n-1}$ \\
\hline
$k$ & $x$ & $x$ & $T^{\overline{x}}_2$ \\ \hline

$k+1$ & $x$ & $\overline{x}$ & $T^{\overline{x}}_1$ \\ \hline

$k+2$ & $\overline{x}\overline{x}x$ & $x$ & $T^x_1$ \\ \hline

$k+4$ & $xx\overline{x}$ & $\overline{x}$ & $T^x_2$ \\ \hline

$k+5$ & $\overline{x}$ & $x$ & $T^x_1$ \\ \hline

$k+6$ & $xx\overline{x}$ & $\overline{x}$ & $T^{\overline{x}}_1$ \\ \hline

$k+8$& $\overline{x}\overline{x}x$ & $x$ & $T^{\overline{x}}_2$ \\ \hline

$k+10$ & $\overline{x}$ & $\overline{x}$ & $T^x_2$\\ \hline

$k+12$ & $x$ & $x$ & $T^{\overline{x}}_4$\\ \hline

$k+13$ & $x$ & $\overline{x}, T^x_2$ & $T^{\overline{x}}_2T^x_2T^{\overline{x}}_1$\\ \hline

$k+14$ & $\overline{x}$ & $x$ & $T^x_1$\\ \hline

$k+16$ & $\overline{x}$ & $\overline{x}$ & $T^x_2$\\ \hline

$k+17$ & $\overline{x}$ & $x$ & $T^x_1$\\ \hline

$k+18$ & $xx\overline{x}$ & $\overline{x}$ & $T^{\overline{x}}_1$\\ \hline

$k+20$ & $\overline{x}\overline{x}x$ & $x$ & $T^{\overline{x}}_2$\\ \hline

$k+21$ & $x$ & $\overline{x}$ & $T^{\overline{x}}_1$\\ \hline

$k+22$ & $\overline{x}\overline{x}x$&$x$&$T^x_1$\\ \hline

$k+24$ & $xx\overline{x}$&$\overline{x}$ & $T^x_2$\\ \hline

$k+26$&$x$&$x$&$T^{\overline{x}}_2$\\ \hline

$k+28$ & $\overline{x}$ & $\overline{x}$ & $T^x_4$\\ \hline

$k+29$ & $\overline{x}$ & $x, T^{\overline{x}}_2$ & $T^x_2T^{\overline{x}}_2T^x_1$\\ \hline

$k+30$ & $x$ & $\overline{x}$ & $T^{\overline{x}}_1$ \\ \hline
\end {tabular}
\end {center}
\vspace*{-3mm}
\end {table}

The construction of the word $S'_{\infty}$, the correctness of which we will prove, is given by Table~\ref{tab1}. According to this table, rule 3 applies to $S_n$ if and only if $P_n$ starts with $x\bar{x}x\bar{x}$. Hence if the word $P_n$ has such a prefix, then $P_{n-1}$ (or $P_{n-2}$ if the $(n{-}1)$th iteration is trivial) has no such prefix; as a result, the word $S_{n-1}$ (respectively, $S_{n-2}$) ends with a 2-block.

Now consider the case $P_n=x\bar{x}x\bar{x}\cdots$ in more details. Without loss of generality, let $P_n$ start with $b$. Then $P_n=babaab\cdots$. Since $P'_n=aabaab\ldots$, the word $S_n$ cannot end with $a$ or with $baab$; thus, it cannot end with a 2-block and we should use rule 3. 

Since $P_n$ is a factor of $\U$ while $\U$ is an infinite product of the blocks $T_2^a=abba$ and $T_2^b=baab$, one of the blocks $T_2^a$ ends in the second position of $P_n$. First consider the following occurrence of $P_n$ in $\U$:
\begin{equation} \label{uuvv}
\U = \cdots \overbrace{abba}^{T_2^a}\lefteqn{\overbrace{\phantom{abba}}^{T_2^a}}ab\underbrace{ba\overbrace{baab}^{T_2^b}\overbrace{baab}^{T_2^b}\cdots}_{P_n}
\end{equation}
Since $P'_{n-1}=bbaab\cdots$, the word $S_{n-1}$ ends with $abba$. Therefore, we cannot put $S_n=ab$ (otherwise $S_n$ will have the suffix $baab$). Further, $P_{n{-}1}$ starts with $abaab$, whence the first letter of $S_n$ is $a$ by (S1). Hence, according to rule 1, the only possibility for $S_n$ is $T^a_2T^b_2T^a_1=abbabaabab$. It is easy to see that $S_{n+1}=ba$ satisfies both (S1) and (S2). 
 
If the last embraced 2-block of (\ref{uuvv}) is $T_2^a$, not $T_2^b$, then we have, up to renaming the letters, the same case as below:
$$
\U = \cdots\overbrace{baab}^{T_2^b}\lefteqn{\overbrace{\phantom{abba}}^{T_2^a}} ab\underbrace{ba\overbrace{baab}^{T_2^b}\cdots}_{P_n}
$$
We assign, as above, $S_n=T^a_2T^b_2T^a_1$ and $S_{n+1}=T^b_1$. The problem appears on the $(n{+}5)$th iteration, because 
$$
P'_{n+4}=\underbrace{\quad b}\underbrace{bab} \underbrace{bab}aab\cdots,
$$
i.e., $S_{n{+}4}$ cannot end with $ba$ or $ab$. Here we have an exclusion from the general method. We use the following trick. At the next three iterations ($(n{+}5)$th to $(n{+}7)$th, the last of them being trivial) we have to add the prefix $baa$ to the fixed context. We will do this prohibiting 3-letter contexts instead of single letters. The word $P_{n{+}3}=babbaba\cdots$ has three left contexts of length 3: $aab$, $baa$, and $bba$. We will prohibit $bba$ on the $(n{+}5)$th iteration and $aab$ on the $(n{+}6)$th one. To do this, we deliberately put $P'_{n+4}=bba\,babbabaab\cdots$, $P'_{n+5}=aab\,babbabaab\cdots$. This allows us to choose $S_{n+4}=ba, S_{n+5}=ab$.

\begin{Rmk}
The above trick leads to one local violation of the general rule on $X_n$. Namely, $|X_{n{+}5}|=n{+}4$ (this word coincides with $X_{n+4}$). The situation is corrected on the next iteration, when we get $|X_{n{+}6}|=n{+}7$ (and the $(n{+}7)$th iteration is trivial).
\end{Rmk}

\begin{Rmk}
The word $T_2^aT_2^aT_2^bT_2^aT_2^a=\theta^2(aabaa)$ is not a factor of $\U$. Hence, the factor $T_2^aT_2^bT_2^a$ occurs in $\U$ inside the factor $T_2^bT_2^aT_2^bT_2^a$ or $T_2^aT_2^bT_2^aT_2^b$. Each such factor requires two uses of the above trick with 3-letter contexts. 
\end{Rmk}

Let us consider the 108-uniform morphism $\psi:\Sigma^*\rightarrow \Sigma^*$, defined by the rules
\begin{subequations}\label{eq:psi}
\begin{align}
\psi(a)&= T_4^aT_2^aT_2^bT_2^aT_4^bT_2^bT_2^aT_4^bT_2^bT_2^aT_2^bT_4^aT_2^aT_2^bT_2^a, \\
\psi(b)&= T_4^bT_2^bT_2^aT_2^bT_4^aT_2^aT_2^bT_4^aT_2^aT_2^bT_2^aT_4^bT_2^bT_2^aT_2^b.
\end{align}
\end{subequations}
Note that the words $\psi(b)$ and $\psi(a)$ coincide up to renaming the letters. A computer check shows that the word $\psi(aabbababbabbaabaababaabb)$ is cube-free. Hence by Theorem \ref {morph}, $\psi$ is a cube-free morphism and the word $\psi(T_{\infty}^b)$ is cube-free. So we put $S'_{\infty}=\psi(T_{\infty}^b)$. The $\psi$-image of one letter equals the product $S_{n-1}S_n\cdots S_{n+30}$ for some number $n$ divisible by 32, see Table~\ref{tab1}. The only exception is described below. Thus, such a $\psi$-image corresponds to 32 successive iterations, during which a 5-block is added to the fixed left context $X_{n-1}$ to get $X_{n+31}$.

There are two different factorizations of the $\psi$-image of a letter, depending on the positions of the factors $T_2^bT_2^aT_2^bT_2^a$ and $T_2^aT_2^bT_2^aT_2^b$ inside and on the borders of the current 5-block of $\U$. These factorizations are presented in the two parts of Table~\ref{tab1}. The mentioned factors occur in the middle of $(2k{+}1)$-blocks for each $k\ge2$. Thus, these factors occur in the middle of each 5-block, and also at the border of two equal 5-blocks. For the latter case, the factorization of the $\psi$-image of the second of two equal letters is given in the righthand part of Table~\ref{tab1}. In the lefthand part of Table~\ref{tab1}, there are two possibilities for $S_{n{+}29}$: the longer [shorter] one should be used if the next 5-block is equal [respectively, not equal] to the current one. In the first case, $S_{n{+}29}$ consists of the last two letters of the $\psi$-image of the current letter and first four letters of the $\psi$-image of the next letter. In the second case, $S_{n+29}$ consists exactly of the two last letters of the $\psi$-image.

The first several iterations are special. Namely, for the regularity of general scheme, we artificially put $W_3=W_0S_{-1}S_1$ (the 1st and the 3rd iterations are trivial by the general condition). 

Thus, we defined the words $S_n$ and then the words $W_n$ for all positive integers $n$. The correctness of the construction is based on the following lemma.

\begin {Lemma} {\label{ml1}}
The word $X_nW_n$ is cube-free for all $n \in \mathbb{N}_0$.
\end {Lemma}

\begin{proof}
We prove by induction that all the words $V_n=(X_nW_nS_nx_n)^3$, where $x_n$ is the letter forbidden on $(n{+}1)$th iteration, have no proper factors that are cubes. This fact immediately implies the statement of the lemma. The inductive base $n\le4$ can be easily checked by hand or by computer. Let us prove the inductive step.The structure of the word $V_n$ is illustrated by the following picture.

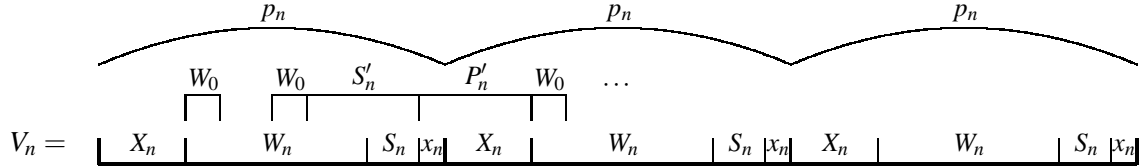
\centerline{ 
\unitlength=1.15mm
\begin{picture}(129,24)(9,2)
\thicklines
\put(18,5){\line(1,0){120}}
\multiput(18,5)(40,0){4}{\line(0,1){3}}
\thinlines
\multiput(28,5)(40,0){3}{\line(0,1){3}}
\multiput(49,5)(40,0){3}{\line(0,1){3}}
\multiput(55,5)(40,0){3}{\line(0,1){3}}
\multiput(38.7,6)(40,0){3}{\makebox(0,0)[cb]{\small$W_n$}}
\multiput(52.2,6)(40,0){3}{\makebox(0,0)[cb]{\small$S_n$}}
\multiput(23.2,6)(40,0){3}{\makebox(0,0)[cb]{\small$X_n$}}
\multiput(56.6,6)(40,0){3}{\makebox(0,0)[cb]{\small$x_n$}}
\put(11,6){\makebox(0,0)[cb]{$V_{n}=$}}
\put(28,13){\line(1,0){4}}
\put(38,13){\line(1,0){34}}
\put(55,13){\line(0,-1){3}}
\multiput(28,13)(4,0){2}{\line(0,-1){3}}
\multiput(38,13)(4,0){2}{\line(0,-1){3}}
\multiput(68,13)(4,0){2}{\line(0,-1){3}}
\multiput(30.2,13.5)(10,0){2}{\makebox(0,0)[cb]{\small$W_0$}}
\put(48.7,13.5){\makebox(0,0)[cb]{\small$S'_n$}}
\put(61.7,13.5){\makebox(0,0)[cb]{\small$P'_n$}}
\put(70.2,13.5){\makebox(0,0)[cb]{\small$W_0$}}
\put(78,13.5){\makebox(0,0)[cb]{\small$\cdots$}}
\qbezier(18,16.5)(38,25)(58,16.5)
\qbezier(58,16.5)(78,25)(98,16.5)
\qbezier(98,16.5)(118,25)(138,16.5)
\multiput(38.2,21.5)(40,0){3}{\makebox(0,0)[cb]{\small$p_n$}}
\end{picture} }

Assume to the contrary that the word $V_n$, $n\ge5$, contains some cube $U^3$. Of course, it is enough to consider the case when the $(n{+}1)$th iteration is nontrivial. The factor $U^3$ of $V_n$ has periods $q=|U|$ and $p_n=|V_n|/3$, but obviously does not satisfy the interaction property. Hence, $|U^3|=3q\le q+p_n-2$ by the Fine and Wilf theorem, yielding $q\le p_n/2-1$. On the other hand, by definition of $W_n$, the longest proper suffix of the word $X_nW_n$ coincides with the longest proper prefix of $V_{n{-}1}$. If $U^3$ contains this prefix, then the latter has periods $q$ and $p_{n-1}=|V_{n-1}|/3$. Applying the Fine and Wilf theorem again, we get $p_{n-1}\le q/2-1$. Excluding $q$ from the two obtained inequalities, we get $p_n\ge 4p_{n-1}+3$. But $p_n=|V_{n-1}|+|S_n|+1\le 3p_{n-1}+17$. Thus, $p_{n-1}\le14$. For $n\ge5$, this is not the case. So, we conclude that $U^3$ does not contain the word $X_nW_n$. \\[5pt]
\textit{Claim 1.} The word $S'_n$ occurs in $V_n$ only three times.
\begin{proof}
Recall that $S'_n$ is a product of 2-blocks (possibly except the last ``odd'' 1-block), and if $n\ge5$, then $S'_n$ begins with a 4-block. Hence, $S'_n$ has no factor $W_0$ and, moreover, cannot begin inside $W_0$. Furthermore, it can be checked by hand or by computer that $S'_{\infty}$ has no Thue-Morse factors of length ${>}48$. Now looking at the structure of $S'_n$ and of $V_n$ one can conclude that any ``irregular'' occurrence of $S'_n$ in $V_n$ should be a prefix of some word $S'_kP'_kW_0$, where $k<n$. The word $S'_k$ is a proper prefix of $S'_n$. The word $P'_k$ is obtained from a Thue-Morse factor by changing the first letter, and hence never begins with a 2-block. Hence, the only possibility is $k=n-1$, and $S_n$ should be the 1-block coinciding with the prefix of $P'_k$. By Table~\ref{tab1}, in all cases when $S_n$ is a 1-block, $P'_{n-1}$ begins with the square of letter, so this possibility cannot take place.  
\end{proof} 
\noindent \textit{Claim 2.} The word $X_nW_nS_nx_n$ is cube-free.
\begin {proof}
The word $X_nW_n$ is a factor of $V_{n-1}$ and hence is cube-free by the inductive assumption. Using again the fact that $S'_n$ is ``almost'' a product of 2-blocks, we conclude that $S'_nx_n$ is also cube-free. So, a cube in $X_nW_nS_nx_n$, if any, contains inside the suffix $S'_{n-1}$ of the word $W_n$. This suffix is preceded by $W_0=aabaaba$; the latter word breaks all periods of $S'_{n-1}$ and does not produce a cube. Hence, the cube should contain more than one occurrence of the factor $S'_{n-1}$. Applying Claim~1 to the words $S'_{n-1}$ and $V_{n-1}$, we see that the cube has the period $p_{n-1}=(|X_nW_n|{+}1)/3$. But this is impossible by condition (S1). The claim is proved.
\end{proof}
Combining Claim~2 with the fact that $U^3$ has no factor $X_nW_n$, we get that $U^3$ is contained inside the word $X_nW_nS_nx_nX_nW_n$. Furthermore, if $S'_n$ is a factor of $U^3$, then the middle occurrence of $U$ is inside $S'_n$ (otherwise, $U^3$ contains one more occurrence of $S'_n$, contradicting Claim~1). In this case, the positions of all factors $aa$ and $bb$ in $U$ have the same parity. But the rightmost occurrence of $U$ in $U^3$ contains a suffix of $S'_n$ followed by a prefix of the word $x_nX_n=P'_n$. The letter $x_n$ breaks this parity of positions, which is impossible. The cases in which all the positions of $aa$ and $bb$ in the rightmost occurrence of $U$ are on the same side of the letter $x_n$, can be easily checked by hand. Thus, we obtain that $S'_n$ is not a factor of $U^3$. Thus, $U^3$ begins inside the factor $S'_nx_n$.

Where the word $U^3$ ends? It is easy to see that the word 
$$
X_nW_n=\bar{x}_{n-1}X_{n-1}W_{n-1}S_{n-1}x_{n-1}X_{n-1}W_{n-1}S_{n-1}x_{n-1}X_{n-1}W_{n-1}S_{n-1}
$$
has the same three occurrences of the factor $S'_{n-1}$ as $V_{n-1}$. So, if $U^3$ contains $S'_{n-1}$, then the middle occurrence of $U$ is inside $S'_{n-1}$. But this is impossible because $S'_{n-1}$ is a rather short suffix of $W_{n-1}$ and the whole word $X_nW_n$ is cube-free. Therefore, $U^3$ should end inside the prefix $\bar{x}_{n-1}X_{n-1}W_{n-1}S_{n-1}$ of $X_nW_n$, like in the following picture.
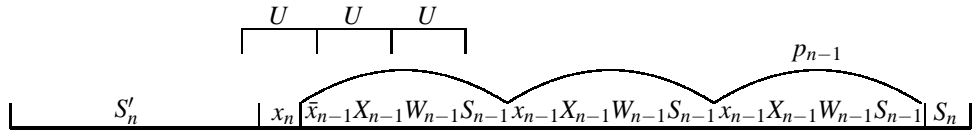
\newline
\unitlength=1.1mm
\centerline{ 
\begin{picture}(120,22)(0,2)
\thicklines
\put(0,6) {\line(1,0){116}}
\put(0,6) {\line(0,1){3}}
\multiput (35,6) (81,0) {2} {\line(0,1) {3}}
\thinlines
\qbezier(35,9)(47.5,17)(60,9)
\qbezier(60,9)(72.5,17)(85,9)
\qbezier(85,9)(97.5,17)(110,9)
\put(97.5,14){\makebox(0,0)[cb]{\small$p_{n-1}$}}
\multiput (30,6) (80.5,0) {2} {\line(0, 1) {3}}
\put (28,18) {\line (1, 0) {27}}
\multiput (28,15) (9, 0) {4} {\line(0, 1) {3}}
\put(14,6.5){\makebox(0,0)[cb]{\small$S'_n$}}
\put(33,6.5){\makebox(0,0)[cb]{\small$x_n$}}
\multiput(32.5,18.5)(9,0){3}{\makebox(0,0)[cb]{\small$U$}}
\put(48,6.5){\makebox(0,0)[cb]{\small$\bar{x}_{n-1}X_{n-1}W_{n-1}S_{n-1}$}}
\multiput(73,6.5)(25,0){2}{\makebox(0,0)[cb]{\small$x_{n-1}X_{n-1}W_{n-1}S_{n-1}$}}
\put(113,6.5){\makebox(0,0)[cb]{\small$S_n$}}
\end{picture}
}
Using the same parity argument as above, we conclude that the word $S'_nx_nX_n=S'_nP'_n$ is cube-free and, moreover, $U^3$ should contain the prefix $aabaa$ of the word $W_{n-1}$. Two cases are to be considered: either $aabaa$ is a factor of $U$ or $aabaa$ occurs in $U^3$ only twice, on the borders of consecutive $U$'s. The second case is impossible, because two closest occurrences of $aabaa$ in $W_{n-1}$ are separated by the factor $babaababbaabbabaabaabb$ which does not contain $P'_n$ as a suffix. For the first case, we get that some (not the leftmost) occurrence of $aabaa$ in $U^3$ is preceded by the concatenation of some suffix of $S'_n$ and the word $P'_n$. If this occurrence of $aabaa$ is a prefix of some $W_0$, then it is preceded by some $P'_k$, $k<n$. But $P'_k$ is not a suffix of $P'_n$, a contradiction. The remaining position for this occurrence of $aabaa$ is the border of some words $S'_k$ and $P'_k$. But then $S'_k$ contains the factor which is on the border between $S'_n$ and $P'_n$, and the parity argument shows that $S'_k$ cannot be partitioned into 2-blocks. This final contradiction shows that $U^3$ cannot be a factor of $V_n$. The lemma is proved.
\end{proof}

\smallskip
By construction, the word $X_n$ is the fixed left extension of $W_n$. Now we consider the second step, that is, the completion of such ``almost uniquely'' extendable word $W_n$ to a premaximal word. The main idea is the same as at the first step. In order to obtain a premaximal word of level $n$, we build the word $W_{n+1}$ in $n{+}1$ iterations by scheme~(\ref {eq:s}) and then prohibit the extension of $W_{n+1}$ by the first letter of the word $P_n$. We denote the obtained premaximal word of level $n$ by $\overline{W}_n$. Then
\begin {equation}
\label {eq:pm}
\overline{W}_n=\underbrace{\phantom{P_n}W_{n+1}\overline{S}_n}\underbrace{P_nW_{n+1}\overline{S}_n}\underbrace{P_nW_{n+1}\overline{S}_n},
\end {equation}
where $\overline{S}_n$ is a ``buffer'' inserted similarly to $S_n$ in order to avoid cubes at the border of the occurrences of $W_{n{+}1}$ and $P_n$. In contrast to the first step, we do not need to build a cube-free right-infinite word, because the construction~(\ref {eq:pm}) is used only once. The form of the word $\overline{S}_n$ depends on the last iteration according to Table~\ref{tab1}; this dependence is described in Table~\ref{tab2}. We choose $\overline{S}_n$ to be the left extension of the word $P_n$ within $\U$ (recall that $P_n{=}\U(n{+}1\ldots1)$).

\begin {table}[!htb]
\setlength {\extrarowheight}{0.8pt}
\vspace*{-3mm}
\caption {the ``final'' suffixes $\overline{S}_n$ for the corresponding iterations from Table~\ref{tab1}. The first column contains the number of the last iteration.}
\label {tab2}
\begin {center}
\begin {tabular} {|l|l|l|}
\hline
Iteration no.  &Prohibitions &  \\
$(n)$ &(Start) &$\overline{S}_{n-1}$ \\
\hline
$k$ &  &\\ \hline

$k+1$ & $\overline{x}$ & $x\overline{x}$  \\ \hline

$k+3$ & $x$ & $\overline{x}$  \\ \hline

$k+4$ & $x$ & $\lambda$ \\ \hline

$k+5$ & $\overline{x}$ & $x\overline{x}\overline{x}x$  \\ \hline

$k+7$ & $\overline{x}$ & $x\overline{x}$  \\ \hline

$k+9$ & $x$ & $\overline{x}x$  \\ \hline

$k+11$ & $\overline{x}$ & $x$  \\ \hline

$k+12$ & $\overline{x}$ & $\lambda$  \\ \hline

$k+13$ & $x$ &  $\lambda$ \\ \hline

$k+15$ & $x$ & $\overline{x}$  \\ \hline

$k+16$ & $x$ & $\lambda$ \\ \hline

$k+18$ & $xx\overline{x}$ & $x\overline{x}$  \\ \hline

$k+19$ &  &   \\ \hline

$k+20$ & $\overline{x}$ & $\lambda$   \\ \hline

$k+23$ & $\overline{x}\overline{x}x$ & $\overline{x}x$  \\ \hline

$k+25$ & $\overline{x}$ & $x\overline{x}$  \\ \hline

$k+27$ & $x$ & $\overline{x}$  \\ \hline

$k+28$ & $x$ & $\lambda$ \\ \hline

$k+29$ & $\overline{x}$ &$\lambda$ \\ \hline

$k+31$ & $\overline{x}$ & $x$  \\ \hline
  
\end {tabular}
\hspace{15 pt}
\begin {tabular} {|l|l|l|}
\hline
Iteration no.  &Prohibitions &  \\
$(n)$ &(Start) &$\overline{S}_{n-1}$ \\
\hline
$k$&$\overline{x}$&$\lambda$\\ \hline

$k+1$ & & \\ \hline

$k+3$ & $\overline{x}\overline{x}x$ & $\overline{x}$ \\ \hline

$k+4$ & $x$ & $\lambda$ \\ \hline

$k+5$ & & \\ \hline

$k+7$ & $xx\overline{x}$ & $x\overline{x}$ \\ \hline

$k+9$ & $x$ & $\overline{x}x$ \\ \hline

$k+11$ & $\overline{x}$ & $x$ \\ \hline

$k+12$ & $\overline{x}$ & $\lambda$ \\ \hline

$k+13$ & $x$ & $\lambda$ \\ \hline

$k+15$ & $x$ & $\overline{x}$ \\ \hline

$k+16$ & $x$ & $\lambda$ \\ \hline

$k+18$ & & \\ \hline

$k+19$ & $xx\overline{x}$ & $x$ \\ \hline

$k+20$ & $\overline{x}$ & $\lambda$ \\ \hline

$k+23$ & $\overline{x}\overline{x}x$ & $\overline{x}x$ \\ \hline

$k+25$ & $\overline{x}$ & $x\overline{x}$ \\ \hline

$k+27$ & $x$ & $\overline{x}$ \\ \hline

$k+28$ & $x$ & $\lambda$ \\ \hline

$k+29$ & $\overline{x}$ & $\lambda$ \\ \hline

$k+31$& $\overline{x}$ & $x\overline{x}$ \\ \hline
 
\end {tabular}
\end {center}
\vspace*{-3mm}
\end {table}

The above idea works without additional gadgets in all cases when $|X_n|=n$. Due to the following obvious remark, it is enough to construct left premaximal words of level $n$ for all $n$ such that $|X_n|=n$; hence, we do not consider constructing the words $\overline{W}_n$ for other values of $n$. 

\begin{Rmk}{\label{Rsc}}
In order to prove the Theorem~\ref{left}, it is sufficient to show the existence of left premaximal words of level $n$ for infinitely many different values of $n$. Indeed, if a word $W$ is left premaximal of level $n$ and $a_1\cdots a_nW$ is a left maximal word, then the word $a_nW$ is left premaximal of level $n{-}1$.
\end {Rmk} 

Using the facts that $W_{n+1}\in\CF$, $\overline{S}_nP_n\in\TM$, and the suffix $S'_n$ of $W_{n{+}1}$ has no long Thue-Morse factors (this is the property of any $\psi$-image), we prove the following lemma. The proof resembles the one of Lemma~\ref{ml1}.

\begin {Lemma} \label{main2}
The word $X_n\overline{W}_n$ is cube-free for all $n \in \mathbb{N}_0$.
\end {Lemma}

Since the word $P_n\overline{W}_n$ is a cube by (\ref{eq:pm}) and at the same time $P_n=X_{n+1}$ is the fixed left context of $W_{n+1}$, we conclude that $X_n$ is the longest left context of the word $\overline{W}_n$. Theorem~\ref{left} is proved.

\begin{Rmk}
For any $n$, the word $\rev(\overline{W}_n)=\overline{W}_n(|\overline{W}_n|)\cdots\overline{W}_n(1)$ is right premaximal of level $n$.
\end{Rmk}

\begin{Rmk}
Our construction provides an upper bound for the length of the shortest left premaximal word of any given level $n$. The results of \cite{Cur} suggest that this length is exponential in $n$. Let $l(n)=|W_n|$. For nontrivial iterations, we have $l(n)=3l(n{-}1)+O(n)$. It is well known that two successive letters in the Thue-Morse word are equal with probability $1/3$. Thus, to obtain $W_n$, we make approximately $2n/3$ nontrivial iterations. So, $l(n)$ is exponential at base $3^{2/3}\approx2.08$. The same  property holds for $|\overline{W}_n|=3l(n{+}1)+O(n)$. It is interesting whether this asymptotics is the best possible.
\end{Rmk}

\begin{proof}[Sketch of the proof of Theorem~\ref{lr}]
Similar to Remark~\ref{Rsc}, it is enough to build premaximal words of level $(n_i,n_i)$ for some infinite sequence $n_1<n_2<\ldots<n_i<\ldots$ of positive integers. We take $n_i=32i+3$ (Table~\ref{tab2} indicates that $\overline{S}_{n_i}=\lambda$, which makes the construction easier). The natural idea is to concatenate left premaximal and right premaximal words through some ``buffer'' word. But we cannot use the words $\overline{W}_n$ for this purpose, because all words $X_n\overline{W}_n$ appear to be right maximal.

So, we modify the last step in constructing left premaximal words as follows. The proof of Lemma~\ref{ml1} implies that the word $X_nW_nS_n\cdots S_{n{+}l}$ is cube-free for any $l$. So, we put
$$
\widetilde{W}_{n_i}=\underbrace{\phantom{P_{n_i}}W_{n_i{+}1}S_{n_i{+}1}S_{n_i{+}2}}\underbrace{P_{n_i}W_{n_i{+}1}S_{n_i{+}1}S_{n_i{+}2}}
\underbrace{P_{n_i}W_{n_i{+}1}S_{n_i{+}1}S_{n_i{+}2}}.
$$ 
By Table~\ref{tab1}, $S_{n_i+3}=\lambda$ and $S_{n_i+4}(1)\ne S_{n_i+1}(1)=x$. The proof of the fact that $X_{n_i}\widetilde{W}_{n_i}\in\CF$ reproduces the proof of Lemma~\ref{main2}. Recall that $S_{n_i+1}(1)=P_{n_i}(1)$ by (S1), yielding that this letter breaks the period of $W_{n_i+1}$ (see (\ref{eq:s1})). On the other hand, the letter $\bar{x}$ breaks the global period of the word $\widetilde{W}_{n_i}$. Hence, the condition $X_{n_i+1}W_{n_i+1}S_{n_i+1}\cdots S_{n_i{+}l}\in\CF$ implies $X_{n_i}\widetilde{W}_{n_i}S_{n_i+3}\cdots S_{n_i{+}l}\in\CF$ for any $l$. Thus, $\widetilde{W}_{n_i}$ is infinitely extendable to the right, left premaximal word of level $n_i$.

Choose an even $m$ such that $|X_{n_i}\widetilde{W}_{n_i}|<2^{m-2}$ and consider the word $\widetilde{W}_{n_i,n_i}=\widetilde{W}_nT_m^{\bar{x}}\rev(\widetilde{W}_n)$:

\centerline{ 
\begin{picture}(95,19)(0,-3)
\gasset{Nw=0.15,Nh=3,Nmr=0,AHnb=0}
\drawline (0, 0) (100,0)
\multiput (0,0) (18,0) {2} {\line (0, 1) {3}} 
\multiput (82,0) (18,0) {2} {\line (0, 1) {3}} 
\drawline (0, 9) (30, 9)
\drawline (70, 9) (100, 9)
\multiput (0,9) (30,0) {2} {\line (0, -1) {3}} 
\multiput (70,9) (30,0) {2} {\line (0, -1) {3}} 
\put (15, 9.5) {\makebox(0,0)[cb]{$\widetilde{W}_{n_i}$}}
\put (85, 9.5) {\makebox(0,0)[cb]{$\rev(\widetilde{W}_{n_i})$}}
\multiput (30, 0) (40, 0) {2} {\line (0, 1) {3}}
\put (16, 1.5) {\makebox (0,0) [cb] {\scriptsize{$W_0$}}}
\put (24, 0.5) {\makebox (0,0) [cb] {\scriptsize{$S'_{n_i+2}$}}}
\put (76, 0.5) {\makebox (0,0) [cb] {\scriptsize{$\rev(S'_{n_i+2}\!)$}}}
\put (50, 0.5) {\makebox (0,0) [cb] {$T_m^{\bar{x}}$}}
\put (-8, 0) {\makebox (0,0) [cb] {$\widetilde{W}_{n_i,n_i}{=}$}}
\end{picture}
}
It remains to prove that the word $X_{n_i}\widetilde{W}_{n_i,n_i}\rev(X_{n_i})$ is cube-free. By the choice of $m$ and overlap-freeness of $T_m^{\bar{x}}$, no cube can contain the factor $T_m^{\bar{x}}$. So, by symmetry, it is enough to check that the word $U=X_{n_i}\widetilde{W}_{n_i}T_m^{\bar{x}}$ is cube-free. Assume to the contrary that it contains a cube $YYY$. Recall that the word $X_{n_i}\widetilde{W}_{n_i}$ is cube-free. Since the first letter of $T_m^{\bar{x}}$ breaks the period of $X_{n_i}\widetilde{W}_n$, one has $|Y|<\per(\widetilde{W}_{n_i})$. Consider the rightmost factor $aabaa$ in $U$; it is inside the factor $W_0$ immediately before the suffix $S'_{n_i{+}2}$ of $\widetilde{W}_n$. If this factor belongs to $YYY$, then $|Y|$ symbols to the left we have another $aabaa$, followed by $S'_{n_i{+}2}$. Then $|Y|=\per(\widetilde{W}_{n_i})$, a contradiction. Hence, $YYY$ has no factors $aabaa$, i.e., is a factor of $abaaba\,S'_{n_i{+}2}T_m^{\bar{x}}$. One can check that the word $S'_{n_i{+}2}$ contains no Thue-Morse factors of length $>48$. The shorter factors can be checked by brute force.

Thus, the word $\widetilde{W}_{n_i,n_i}$ is premaximal of level ($n_i,n_i$). The theorem is proved.
\end{proof}

\nocite{*}
\bibliographystyle{eptcs}

\end{document}